\documentclass[12pt,a4paper,headings=small,oneside]{scrartcl}

\usepackage{amsmath, amssymb, amsfonts}
\usepackage{amsthm}
\usepackage{mathtools}
\usepackage{graphics}
\usepackage{epstopdf}
\usepackage[round]{natbib}
\usepackage{bbm}
\usepackage{bm}
\usepackage{setspace}
\usepackage{color}
\usepackage[hyphens]{url}
\usepackage{hyperref}

\newtheorem{theorem}{Theorem}

\newtheorem{proposition}[theorem]{Proposition}

\theoremstyle{definition}

\newtheorem{example}[theorem]{Example}

\newtheorem{remark}[theorem]{Remark}

\newcommand{\cd}{\stackrel{\mathcal{D}}{\longrightarrow}}

\newcommand{\Ex}{\mathbb{E}} 
\newcommand{\Var}{\operatorname{Var}} 
\newcommand{\Cov}{\operatorname{Cov}} 
\newcommand{\ind}{\mathbbm{1}} 

\begin{document}

\title{A gamma tail statistic and its asymptotics}

\author{Toshiya Iwashita$^{1}$\footnote{iwashita\_toshiya@rs.tus.ac.jp} {} and Bernhard Klar\,$^{2}$\footnote{ bernhard.klar@kit.edu, corresponding author} \\
    \small{$^{1}$ Institute of Arts and Sciences, Tokyo University of Science} \\
    \small{$^{2}$ Institute of Stochastics, Karlsruhe Institute of Technology (KIT)}}

\date{\today }
\maketitle

\begin{abstract}
Asmussen and Lehtomaa [Distinguishing log-concavity from heavy tails. Risks 5(10), 2017] introduced an interesting function $g$ which is able to distinguish between log-convex and log-concave tail behaviour of distributions, and proposed a randomized estimator for $g$.
In this paper, we show that $g$ can also be seen as a tool to detect gamma distributions or distributions with gamma tail. We construct a more efficient estimator $\hat{g}_n$ based on $U$-statistics, propose several estimators of the (asymptotic) variance of $\hat{g}_n$, and study their performance by simulations. Finally, the methods are applied to several data sets of daily precipitation.
\end{abstract}


\section{Introduction} \label{section1}

Throughout the paper, we consider independent and identically distributed (i.i.d.) random variables $X,X_1,X_2,\ldots>0$ with common distribution function $F$ having density $f$.
\cite{AL:2017} introduced the function $g:(0,\infty)\to [0,1]$, defined by
\[
g_X(d)= g(d)= \Ex \left[ \frac{|X_1-X_2|}{X_1+X_2} \Big |X_1+X_2>d \right].
\]
To start with, note that the function $g$ satisfies $g_{aX}(d)=g_X(d/a)$ for $a>0$. Hence, a rescaling of $X$ does not change the qualitative behaviour of $g$.
The function has the following interpretation \citep{AL:2017}: If both $X_1$ and $X_2$ contribute equally to the sum $X_1+X_2$, then $g$ should eventually obtain values close to 0; if only one of the variables tends to be of the same magnitude as the whole sum, then $g$ is close to 1 for large $d$. More formally, they showed that $g(d)\to 1$ for $d\to\infty$ for many distributions with long tails, e.g. for lognormal type distributions, for Weibull distributions with shape parameter $\alpha<1$, and regularly varying distributions $RV(\alpha)$ with $\alpha>1$ and eventually decreasing density $f$. Here, a property holds {\em eventually}, if there exists $x_0$ so that the property holds in the set $[x_0,\infty)$. Further literature related to the single-big-jump principle is \cite{BB:2015} and \cite{LE:2015}.

A density $f$ is called {\em log-concave}, if $f(x)=e^{\phi(x)}$, where $\phi$ is a concave function. If $\phi$ is convex, then $f$ is log-convex. \cite{AL:2017} proved the following result. Assume that the density $f$ is  twice differentiable and eventually log-concave. Then,
\[
\limsup_{d\to\infty} g(d) \leq 1/2.
\]
Similarly, if $f$ is eventually log-convex, then
$\liminf_{d\to\infty} g(d) \geq 1/2.$
Moreover, the proof of Theorem 1 in \cite{AL:2017} shows that $g(d)\leq 1/2$ for all $d>0$, if $f$ is log-concave and twice differentiable. If $f$ is log-convex, $g(d)\geq 1/2$ for all $d>0$.
Since the exponential distribution is log-concave and log-convex, it follows that $g(d)=1/2$ for all $d>0$ under exponentiality.

A gamma distribution with density $f(x)=\beta^{\alpha} x^{\alpha-1}\exp(-\beta x)/\Gamma(\alpha)$, where shape parameter $\alpha$ and rate $\beta$ (or scale parameter $1/\beta$) are positive, is log-concave for $\alpha\geq 1$. Hence $g(d)\leq 1/2$ for all $d>0$. Similarly, for $\alpha\leq 1$, it is log-convex, and we have $g(d)\geq 1/2$ for $d>0$. Our first result in Sec. \ref{sec2} shows that $g(d)$ takes a constant value for gamma distributions; moreover, the family of gamma distributions is characterized by this property. Hence, $g(d)$ can also be seen as a tool to detect gamma distributions or distributions with gamma tail.

In Sec. \ref{sec3}, we first analyze the asymptotic behaviour of a randomized estimator of $g(d)$ introduced by \cite{AL:2017}, and construct a more efficient estimator based on $U$-statistics, denoted by $\hat{g}_n(d)$. In Sec. \ref{sec4} and \ref{sec5}, we propose several estimators of the (asymptotic) variance of $\hat{g}_n(d)$ and study their performance by simulations.
Finally, in Sec. \ref{sec6}, the methods are applied to several data sets of daily point and areal precipitation.

\section{Properties of function \texorpdfstring{\boldmath $g$}{g}} \label{sec2}

Our first result is based on Lukacs' Theorem \citep{LU:1954}, which states the following: Let $X$ and $Y$ be positive and independent random variables. Then $U=X+Y$ and $V=X/Y$ are independent if and only if both $X$ and $Y$ have gamma distributions with the same scale parameter.

\begin{proposition} \label{g-gamma}
\begin{enumerate}
\item[a)]
Let $X$ and $Y$ be positive and independent random variables. Then,
\begin{align*}
\Ex \left[ \frac{X-Y}{X+Y} \Big |X+Y>d \right]
&= \Ex \left[ \frac{X-Y}{X+Y} \right], \quad \text{for all } d>0,
\end{align*}
if and only if both $X$ and $Y$ have gamma distributions with the same scale parameter.
\item[b)]
Assume that $X,X_1,X_2$ are i.i.d. random variables. Then,
\begin{align*}
g_{X}(d) &= g_X(0), \quad \text{for all } d>0,
\end{align*}
if and only if $X$ is gamma distributed.
\end{enumerate}
\end{proposition}

\begin{proof}
Let $X$ and $Y$ be positive and independent random variables. Then, using Lukacs Theorem, $h_1(V)=(1+1/V)^{-1}=X/(X+Y)$ and $X+Y$ are independent, if and only if  both $X$ and $Y$ have gamma distributions with the same scale parameter,  and the same assertion holds for $h_2(V)=(1+V)^{-1}=Y/(X+Y)$.
Since the function $h_1(v)-h_2(v)$ is stricly increasing for $v>0$, the independence condition is equivalent to the condition that $R=(X-Y)/(X+Y)=h_1(V)-h_2(V)$ and $X+Y$ are independent, or, likewise, to the condition
\begin{align*}
  \Ex[ R | X+Y>d ] &= \Ex[R], \quad \text{for all } d>0.
\end{align*}
This proves part a). Now, additionally assume that $X$ and $Y$ have the same distribution. Then, the distribution of $R$ is symmetric around 0. Hence, the sigma algebras generated by $R$ and $|R|$ coincide, which yields the assertion in b).
\end{proof}

\begin{remark} \label{rem-gamma}
From Proposition \ref{g-gamma} and the remarks in Section \ref{section1}, we obtain $g(d)=c(\alpha)\leq 1/2$ for all $d>0$, if $\alpha\geq 1$ . Similarly, for $\alpha\leq 1$, we have $g(d)=c(\alpha)\geq 1/2$ for all $d>0$.
From Prop. \ref{prop} in Appendix \ref{appendix}, we obtain the explicit values
\[
c(\alpha) = \left(2^{2\alpha-1} \, \alpha \, \mathrm{B}(\alpha,\alpha)\right)^{-1},
\]
where $\mathrm{B}(\cdot,\cdot)$ denotes the beta function.
For $\alpha=1/5$ and $\alpha=5$, we get
$c(1/5)\approx 0.798$ and $c(5)=63/256 \approx 0.246$, respectively.
These results show formally what can be seen in the left and right panels of Figure 1 in \cite{AL:2017}, which are generated using simulated data.

A typical measure to describe the tail of loss distributions is the asymptotic behaviour of the failure rate \citep[p. 34]{KP:2012}
\[
h(\infty) = \lim_{x\to\infty} h(x)
= \lim_{x\to\infty} \frac{f(x)}{1-F(x)}
= - \lim_{x\to\infty} \frac{d}{dx}\log f(x),
\]
where the last equality holds for distributions with support $[0,\infty)$.
 For gamma distributions with scale parameter 1, one has $h(\infty)=1$, irrespective of $\alpha$. Hence, this measure is not able to distinguish between gamma distributions with different shape parameters, in contrast to the function $g$. The same holds for the limit of the mean excess function \citep[p. 35]{KP:2012}. Both failure rate and mean excess function are nonlinear for gamma distributions.
\end{remark}

\section{A new proposal for an estimator of \texorpdfstring{\boldmath ${g(d)}$}{g(d)}} \label{sec3}

\subsection{Asymptotic behaviour of the Asmussen-Lehtomaa estimator} \label{sec3-1}
To estimate $g(d)$ based on an i.i.d. sample $X_1,\ldots,X_n$, where $n=2m$ is even, \cite{AL:2017} proposed the following estimator: use any pairing $(Y_k,Z_k)_{1\leq k\leq m}$ of the $X_i$ (e.g., $Y_k=X_{2k-1}, Z_k=X_{2k}$), and set
\begin{align} \label{AL-estimator}
  \tilde{g}_m(d) &= \frac{\sum_{k=1}^{m} R_k \ \ind\left(Y_k+Z_k>d\right)}{\sum_{k=1}^{m} \ind\left(Y_k+Z_k>d\right)},
  \quad \text{with } R_k=\frac{|Y_k-Z_k|}{Y_k+Z_k},
\end{align}
and where $\ind(A)$ is the indicator function of the event $A$.
The estimator proposed in (\ref{AL-estimator}) has the advantage that it can be computed fast even for very large sample sizes. On the other hand, it doesn't make efficient use of the sample; moreover, it requires splitting the sample randomly in two halves, leading to a randomized statistic.
This is illustrated in Figure \ref{fig1}; see also Figure 3 in \cite{AL:2017}.
Since we are particularly interested in the tail behaviour, i.e. in large values of $d$, the sample size will typically be small, and the disadvantages predominate.

\begin{figure}
  \centering
  \includegraphics[scale=0.7]{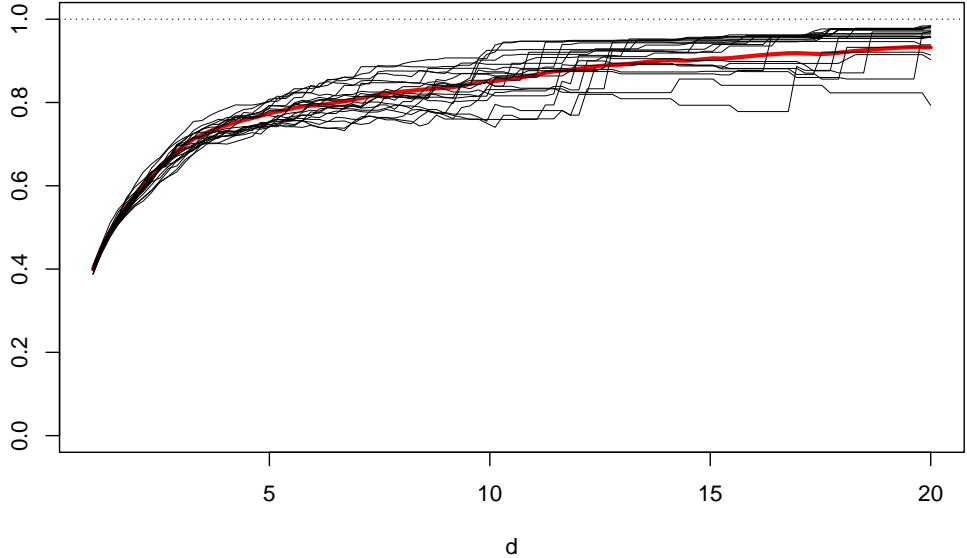}
  \caption{In black: Graphs of 20 versions of $\tilde{g}_m(d)$, generated from random partitions of a classical set of Danish fire insurance data, available in the R package {\em evir} \citep{EV:2018}. The dataset is scaled to have mean 1. In red: Graph of the new estimator $\hat{g}_n(d)$. \label{fig1}}
\end{figure}

\medskip
To derive the limiting distribution of $\tilde{g}_m(d)$, define for any $d\geq 0$ such that $\nu_d=P(Y_1+Z_1>d)>0$ the quantities $S_k= R_k\, \ind(Y_k+Z_k>d), T_k=\ind(Y_k+Z_k>d)$, $\mu_d=\Ex(S_1)>0$, $\mu_{2,d}=\Ex(S_1^2)$. By the central limit theorem,
\begin{align*}
\frac{1}{\sqrt{m}} \sum_{k=1}^{m} \left(
\begin{pmatrix} S_k \\ T_k \end{pmatrix} -
\begin{pmatrix} \mu_d \\ \nu_d \end{pmatrix} \right)
& \cd N_2( \mathbf{0}, \tilde{\Sigma}_d),
\quad \text{where } \tilde{\Sigma}_d =
\begin{pmatrix}
\mu_{2,d}-\mu_d^2 & \mu_d(1-\nu_d) \\
\mu_d(1-\nu_d) & \nu_d(1-\nu_d)
\end{pmatrix}.
\end{align*}
Then, the delta method yields
\[
\sqrt{m} \left( \tilde{g}_m(d) - g(d) \right)
\cd N\left(0,
\frac{\mu_{2,d}}{\nu_d^2} - \frac{\mu_d^2}{\nu_d^3}\right).
\]
Noting that $g(d)=\mu_d/\nu_d$ and writing $g_2(d)=\Ex[R_1^2|Y_1+Z_1>d]=\mu_{2,d}/\nu_d$, we end up with
\begin{align} \label{asymp-g-tilde}
\sqrt{m\nu_d} \left( \tilde{g}_m(d) - g(d) \right)
& \cd N\left(0, g_2(d) - \left(g(d)\right)^2 \right).
\end{align}
Since $R_k<1$, this result holds without any assumptions, as long as $\nu_d>0$. Note that the effective sample size is $m\nu_d$.

\begin{example}
Assume that $Y_1$ and $Z_1$ are i.i.d. gamma-distributed with shape parameter $\alpha$ and rate $\beta$. Then, $\nu_d=P(W>d)$, where $W$ has a gamma distribution with parameters $2\alpha$ and $\beta$. Using Proposition \ref{g-gamma} and Prop. \ref{prop} in Appendix \ref{appendix}, we obtain
\[
g_2(d) \,=\, \Ex[R_1^2]
\,=\, \frac{\Gamma(2\alpha+1)}{\Gamma(2\alpha+2)}
\]
for all $d>0$. It follows that
\begin{align*}
\sqrt{m\nu_d} \left( \tilde{g}_m(d) - {c(\alpha)} \right) & \cd
N\left(0, \frac{\Gamma(2\alpha+1)}{\Gamma(2\alpha+2)} -c^2(\alpha) \right),
\end{align*}
where $c(\alpha)$ is given in Remark \ref{rem-gamma}. For $\alpha=1$, i.e. the exponential distribution, this results in
\[
\sqrt{m\nu_d} \left( \tilde{g}_m(d) - 1/2 \right) \cd
N\left(0, 1/12 \right).
\]
\end{example}

\subsection{A new estimator based on \texorpdfstring{\boldmath $U$-statistics}{U-statistics}}
A more efficient way of estimating $g(d)$ is the use of suitable $U$-statistics (for the general theory, see \cite{KB:1994,LE:1990}). To this end, define kernels of degree 2
\begin{align*}
  h^{(1)}(x_1,x_2;d) = \frac{|x_1-x_2|}{x_1+x_2} \, \ind\left(x_1+x_2>d\right),  &\quad
  h^{(2)}(x_1,x_2;d) = \ind\left(x_1+x_2>d\right),
\end{align*}
and define two $U$-statistics by
\begin{align*}
  U_n^{(l)}(d) &=
  \frac{2}{n(n-1)} \sum_{1\leq i<j\leq n} h^{(l)}(X_i,X_j;d), \quad l=1,2.
\end{align*}
Obviously, $U_{n}^{(l)}(d)$ is an unbiased estimator of $\theta^{(l)}_d=\Ex[h^{(l)}(X_1,X_2;d)], l=1,2$. Note that $\theta^{(1)}_d$ and $\theta^{(2)}_d$ coincide with $\mu_d$ and $\nu_d$.
Then, estimate $g(d)$ by the ratio of these statistics:
\begin{align} \label{U-est}
  \hat{g}_n(d) &= \frac{U_n^{(1)}(d)}{U_n^{(2)}(d)}, \quad d>0.
\end{align}
By the strong law of large numbers for $U$-statistics \citep[p.~122]{LE:1990}, $U_{n}^{(l)}(d), l=1,2,$ and hence $\hat{g}_n(d)$ are strongly consistent estimators for $\theta^{(l)}_d$ and $g(d)$, respectively.

\smallskip
The joint asymptotic distribution of $U$-statistics can be found in  \citet[p.~76]{LE:1990} or \citet[p.~132]{KB:1994}. This yields the following result.

\begin{proposition} \label{prop4}
For $l=1,2$, let
\begin{align*}
\psi^{(l)}(x_1,x_2;d) &= h^{(l)}(x_1,x_2;d)-\theta^{(l)}_d, \\
\psi^{(l)}_1(x_1;d) &= \Ex\big[\psi^{(l)}(x_1,X_2;d)\big].
\end{align*}
Further, define
\begin{align*}
\eta^{(l)}_1(d)=\Ex\left[\left( \psi^{(l)}_1(X_1;d) \right)^2\right] \ (l=1,2), \quad
\eta^{(1,2)}_1(d)=\Ex\left[\psi^{(1)}_1(X_1;d)\, \psi^{(2)}_1(X_1;d)\right].
\end{align*}
If $\eta^{(l)}_1(d)>0$ for $l=1,2$, then,
\begin{align*}
\sqrt{n} \left(
\begin{pmatrix} U_{n}^{(1)}(d) \\ U_{n}^{(2)}(d) \end{pmatrix} -
\begin{pmatrix} \mu_d \\ \nu_d \end{pmatrix} \right)
& \cd N_2(\mathbf{0},4\Sigma_d),
\quad \text{where } \Sigma_d =
\begin{pmatrix}
\eta^{(1)}_1(d) & \eta^{(1,2)}_1(d) \\
\eta^{(1,2)}_1(d) & \eta^{(2)}_1(d)
\end{pmatrix}.
\end{align*}
\end{proposition}

Using Prop. \ref{prop4} and the delta method, we can derive the asymptotic behaviour of $\hat{g}_n(d)$.

\begin{theorem} \label{theorem4}
Let $\nu_d>0$, and $\eta^{(l)}_1(d)>0$ for $l=1,2$. Then,
\begin{align} \label{asymp-g-hat}
\sqrt{n\nu_d} \left( \hat{g}_n(d) - g(d) \right)
& \cd N\left(0, \sigma_d^2 \right),
\end{align}
where
\begin{align} \label{sigma-d}
\sigma_d^2 = \frac{4}{\nu_d} \left( \eta^{(1)}_1(d)
- 2g(d) \, \eta^{(1,2)}_1(d) + g^2(d) \, \eta^{(2)}_1(d) \right).
\end{align}
\end{theorem}

\subsection{Asymptotic relative efficiency}

Comparing (\ref{asymp-g-tilde}) with (\ref{asymp-g-hat}), one may anticipate that the asymptotic relativ efficiency (ARE) of $\tilde{g}_m(d)$ relativ to $\hat{g}_n(d)$, where $n=2m$, is roughly 1/2; here, the ARE is given by
\begin{align*}
ARE(\tilde{g}_m(d),\hat{g}_n(d)) &=\frac{\sigma_d^2}{2\tilde{\sigma}_d^2},
\end{align*}
where $\tilde{\sigma}_d^2=g_2(d)-(g(d))^2$.
Figure \ref{fig2} shows the (numerically computed) ARE's for gamma distributions with various shape parameters $\alpha$ and rate $\beta=\alpha$ (such that the expectation is 1) for $0<d\leq 5.2$.

\begin{figure}
  \centering
  \includegraphics[scale=0.65]{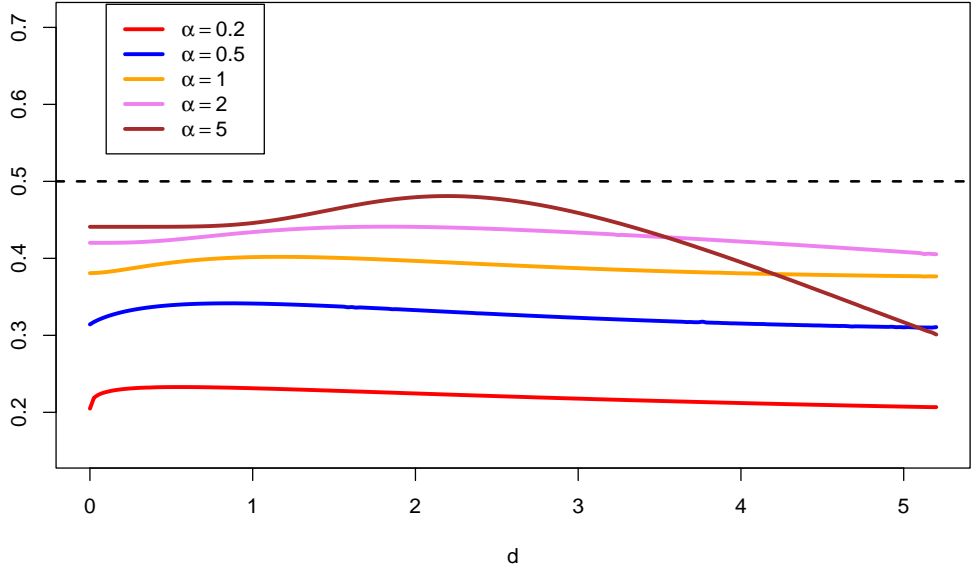}
  \caption{The graph shows the asymptotic relativ efficiency of $\tilde{g}_m(d)$ relativ to $\hat{g}_n(d)$ for gamma distributions with different shape parameters $\alpha$ and rate $\beta=\alpha$. \label{fig2}}
\end{figure}

First, note that $\sigma_d^2$, unlike $\tilde{\sigma}_d^2$, depends on $d$ even for the gamma distribution. For all cases considered, the ARE  of $\tilde{g}_m(d)$ relativ to $\hat{g}_n(d)$ is smaller than 0.5. By and large, the ARE is increasing in the shape parameter: it is between $0.20$ and $0.23$ for $\alpha=0.2$, between $0.38$ and $0.40$ for the exponential distribution, and ranges from $0.41$ to $0.44$ for $\alpha=2$.
For fixed $\alpha$ and large values of $d$, the ARE decreases markedly, which can be observed in Fig. \ref{fig2} for $\alpha=5$, but also occurs for other values of $\alpha$ for larger $d$.
In summary, it becomes apparent that the estimator based on the ratio of $U$-statistics is much more efficient than the proposal in Section \ref{sec3-1}.


\section{Estimators of variance} \label{sec4}

In this section, we discuss and compare several methods of estimating the variance $\sigma_d^2$ in (\ref{sigma-d}) or its counterpart for finite sample size. There are at least three general approaches. The first one is to derive consistent estimators $\sigma_{U^{(l)}}^2 (l=1,2)$ and $\sigma_{U^{(1,2)}}^2$ of $\Var(U_{n}^{(l)}(d))$ and  $\Cov(U_{n}^{(1)}(d), U_{n}^{(2)}(d))$, respectively. Then, a consistent estimator of $\sigma_d^2$ is given by
\begin{align} \label{hat-sigma-d}
\hat{\sigma}_d^2 = \frac{n}{U_{n}^{(2)}(d)} \left(
\sigma_{U^{(1)}}^2 - 2\hat{g}_n(d) \, \sigma_{U^{(1,2)}}^2
+ \hat{g}_n^2(d) \, \sigma_{U^{(2)}}^2 \right).
\end{align}
Second, one can estimate the quantities $\eta^{(l)}_1(d) (l=1,2)$ and $\eta^{(1,2)}_1(d)$ in the asymptotic covariance matrix, and use (\ref{sigma-d}), with $\nu_d$ replaced by $U_n^{(2)}$.
The third possibility is a direct approach using resampling procedures.

\subsection{Variance estimation using the unbiased variance estimator}

Let $U_n=2/(n(n-1)) \sum_{i<j} h(X_i,X_j)$ be a general $U$-statistic of degree 2, estimating $\theta=\Ex h(X_1,X_2)$. Defining $\zeta_0=\theta^2$, $h_1(x_1)=\Ex h(x_1,X_2)$ and
\begin{align*}
  \zeta_1=\Ex\left[h_1^2(X_1)\right],  &\quad
  \zeta_2=\Ex\left[h^2(X_1,X_2)\right],
\end{align*}
the finite sample variance of $U_n$ is given by
\begin{align*}
  \Var(U_n) &= \frac{2}{n(n-1)} \left\{ 2(n-2)\zeta_1 + \zeta_2 - (2n-3) \zeta_0 \right\}.
\end{align*}
One can estimate $\zeta_c, c=0,1,2,$ by
\begin{align}
\hat{\zeta}_0 &= \frac{1}{n^{\underline{4}}} \sum_{d(i,j,k,l)} h(X_i,X_j)h(X_k,X_l), \nonumber \\
\hat{\zeta}_1 &= \frac{1}{n^{\underline{3}}} \sum_{d(i,j,k)} h(X_i,X_j)h(X_i,X_k), \qquad
\hat{\zeta}_2 = \binom{n}{2}^{-1} \sum_{i<j} h^2(X_i,X_j), \label{zeta1}
\end{align}
where $d(i_1,\ldots,i_m)$ denotes a set of distinct indices $1\leq i_1,\ldots,i_m \leq n$, and $n^{\underline{m}}=n(n-1)\cdots (n-m+1)$.
Then, the minimum variance unbiased estimator of $\Var(U_n)$ is given by \citep{SS:1992}
\begin{align} \label{unbiased-est}
  \hat{\sigma}_U^2 &= \frac{2}{n(n-1)} \left\{ 2(n-2)\hat{\zeta}_1 + \hat{\zeta}_2 - (2n-3) \hat{\zeta}_0 \right\}
  \ = \ U_n^2 - \hat{\zeta}_0.
\end{align}
The second equality has also been noted by \citet{WL:2014}. All formulas can directly be generalized to multivariate $U$-statistics by writing $hh^T$ and $h_1h_1^T$ instead of $h^2$ and $h_1^2$.
The degree of the $U$-statistics in (\ref{unbiased-est}) is 4. To reduce the computational burden, it is possible to rewrite it as
\begin{align} \label{unbiased-est2}
  \hat{\sigma}_U^2 &= \frac{4C_1^2-2C_2^2}{n^{\underline{4}}}  - \frac{4n-6}{(n-2)(n-3)} U_n^2,
\end{align}
where
\begin{align*}
C_1^2 &= \sum_{i=1}^n \bigg( \sum_{j\neq i} h(X_i,X_j) \bigg)^2, \qquad
C_2^2 = \sum_{i\neq j} h^2(X_i,X_j)
\end{align*}
\citep[p. 2972]{SS:1992}. In (\ref{unbiased-est2}), the number of summands is $O(n^2)$ compared to $O(n^4)$ in (\ref{unbiased-est}).
To obtain a multivariate version of (\ref{unbiased-est2}), write $hh^T$ and $U_nU_n^T$ instead of $h^2$ and $U_n^2$, and define
\begin{align} \label{unbiased-est3}
C_1^2 &= \sum_{i=1}^n S_iS_i^T, \quad \text{where }
S_i=\sum_{j\neq i} h(X_i,X_j).
\end{align}
After plugging the unbiased estimator of the covariance matrix in formula (\ref{hat-sigma-d}), we denote the resulting estimator by $\hat{\sigma}_{U,d}^2$.

\subsection{Estimating the variance using Noether's estimator}

Two further proposals given in \citet{SS:1992} for estimating the variance of a U-statistic of degree 2 are the Noether and modified Noether estimator defined by
\begin{align*} 
  \hat{\sigma}_N^2 &= \binom{n}{2}^{-2} C_1^2
  - \binom{n}{2}^{-1} \Big\{ (2n-3) U_n^2 +1 \Big\}, \\
  \hat{\sigma}_{Nm}^2 &= \frac{n(n-1)}{(n-2)(n-3)} \hat{\sigma}_N^2.  \label{noether-mod}
\end{align*}
Computing multivariate generalizations and plugging them in (\ref{hat-sigma-d}) leads to estimators denoted $\hat{\sigma}_{N,d}^2$ and $\hat{\sigma}_{Nm,d}^2$.

\subsection{Variance estimation using the large-sample variance}

This approach uses plug-in estimators for $\eta^{(l)}_1(d) (l=1,2)$ and $\eta^{(1,2)}_1(d)$ given in Proposition \ref{prop4}.
Define $\hat{\zeta}_1^{(l)}$ as $\hat{\zeta}_1$ in (\ref{zeta1}), replacing $h$ by $h^{(l)}$, and put
\begin{align*}
\hat{\zeta}_1^{(1,2)} &= \frac{1}{n^{(3)}} \sum_{d(i,j,k)}  h^{(1)}(X_i,X_j) \; h^{(2)}(X_i,X_k).
\end{align*}
Then, the estimators for the entries in the large-sample covariance matrix are given by
\begin{align*}
  \hat{\eta}^{(l)}_1(d) &= \hat{\zeta}_1^{(l)} - \big( U_{n}^{(l)}(d)\big)^2,  \quad l=1,2, \\
 \hat{\eta}^{(1,2)}_1(d) &= \hat{\zeta}_1^{(1,2)} - U_{n}^{(1)}(d)U_{n}^{(2)}(d).
\end{align*}
Replacing all quantities in (\ref{sigma-d}) by the corresponding estimators yields the variance estimator $\hat{\sigma}_{L,d}^2$.
By formula (8) in \cite{SS:1992}, one has
$\hat{\zeta}_1 = (C_1^2-C_2^2)/n^{\underline{3}}$, and using (\ref{unbiased-est3}), we obtain a corresponding multivariate generalization with complexity $O(n^2)$.

\subsection{Variance estimators based on resampling procedures}

To obtain the bootstrap estimator, let
$(X_{j1}^*,\ldots,X_{jn}^*), j=1,\ldots,M,$ be conditionally independent samples with distribution function $F_n$, given $X_1,\ldots,X_n$. Here, $F_n$ denotes the empirical distribution function of $X_1,\ldots,X_n$.
For the statistic $\hat{g}_n(d)=g( U_n^{(l)}(d), l=1,2)$ in (\ref{U-est}), one has to compute
\begin{align*}
g_n^{(j)}(d) &= g\left(
U_n^{(l)}(X_{j1}^*,\ldots,X_{jn}^*;d), l=1,2 \right)
\end{align*}
for $j=1,\ldots,M$, and
$\bar{g}_n^*(d)=M^{-1}\sum_{j=1}^M g_n^{(j)}(d)$.
Then, the Monte Carlo version of the bootstrap estimator of $\Var(\hat{g}_n(d))$ is given by
\begin{align*}
\hat{\sigma}_{B,d}^2 &= \frac{1}{M-1} \sum_{j=1}^{M}
\left( g_n^{(j)}(d) - \bar{g}_n^*(d) \right)^2.
\end{align*}
The number of bootstrap replications $M$ should not be chosen too small; we use $M=999$ in all simulations in the next section.

\medskip
The jackknife procedure for a function of several $U$-statistics is described in \citet[p. 227]{LE:1990}. Here, we compute
\begin{align*}
g_{n-1}^{(-j)}(d) &= g\left(
U_{n-1}^{(l)}(X_1,\ldots,X_{j-1},X_{j+1},\ldots,X_n;d), l=1,2 \right)
\end{align*}
for $j=1,\ldots,n$, and
$\bar{g}_{n-1}(d)=n^{-1}\sum_{i=1}^n g_{n-1}^{(-j)}(d)$. The jackknife estimator of $\Var(\hat{g}_n(d))$ is given by
\begin{align*}
\hat{\sigma}_{J,d}^2 &= \frac{n-1}{n} \sum_{j=1}^{n}
\left( g_{n-1}^{(-j)}(d)  - \bar{g}_{n-1}(d) \right)^2.
\end{align*}
\cite{CV:1981} show that the jackknife estimator of the variance of a U-statistic with degree 2 has some desirable properties. The comparative performance of $\hat{\sigma}_{J,d}^2$ is examined in the next section via simulations.

\section{Numerical illustrations} \label{sec5}
\subsection{RSME and bias of the different estimators of variance} \label{sec5-1}
In the first part of this section, we compare the performance of all estimators of the variance of $\hat{g}_n(d)$ introduced in Sec. \ref{sec4} by computer simulations. Hence, in a first step, we approximated the true variance of $\sqrt{n\nu_d}\,\hat{g}_n(d)$ by a Monte Carlo simulation with $10^6$ replications.
In a second simulation with $10^4$ repetitions, the averages (Ave) of the relative values (i.e. estimator divided by the true variance) and the root mean squared error (RMSE) are computed.

As distributions, we choose gamma distributions with different shape parameter $\alpha$ and rate $\beta=\alpha$. In all simulations, we use effective sample sizes, defined as follows: for given values of $\alpha$ and $d$, the total sample size $n$ was chosen such that $n\nu_d= nP(X_1+X_2>d)=n_{\text{eff}}$.
Tables \ref{tab1}-\ref{tab3} show the results.

\begin{table}
\centering
\begin{tabular}{rlrrrrrr}
  \hline
 $\alpha$ &      & $\hat{\sigma}_{U,d}^2$ & $\hat{\sigma}_{N,d}^2$ & $\hat{\sigma}_{Nm,d}^2$ & $\hat{\sigma}_{L,d}^2$ & $\hat{\sigma}_{B,d}^2$ & $\hat{\sigma}_{J,d}^2$ \\
  \hline
  0.2 & Ave & 0.893 & 0.776 & 0.810 & 0.822 & 1.143 & 0.976 \\
   & RMSE & 0.017 & 0.016 & 0.016 & 0.017 & 0.023 & 0.019 \\
  0.5 & Ave & 0.899 & 0.670 & 0.701 & 0.838 & 1.028 & 0.993 \\
   & RMSE & 0.020 & 0.027 & 0.025 & 0.021 & 0.021 & 0.022 \\
  1.0 & Ave & 0.886 &   -   &   -   & 0.837 & 0.972 & 0.996 \\
   & RMSE & 0.017 &   -   &   -   & 0.018 & 0.016 & 0.018 \\
  2.0 & Ave & 0.882 &   -   &   -   & 0.846 & 0.941 & 1.011 \\
   & RMSE & 0.014 &   -   &   -   & 0.015 & 0.014 & 0.015 \\
  5.0 & Ave & 0.879 &   -   &   -   & 0.862 & 0.912 & 1.022 \\
   & RMSE & 0.010 &   -   &   -   & 0.010 & 0.010 & 0.012 \\
   \hline
\end{tabular}
\caption{Ave and RMSE of all estimators introduced in Sec. \ref{sec4} for $n_{\text{eff}}=20$, $d=3$ and varying shape parameter $\alpha$. The entry - indicates a negative value in more than 1\% of cases.} \label{tab1}
\end{table}

\begin{table}
\centering
\begin{tabular}{rlrrrrrr}
  \hline
 $n_{\text{eff}}$ &      & $\hat{\sigma}_{U,d}^2$ & $\hat{\sigma}_{N,d}^2$ & $\hat{\sigma}_{Nm,d}^2$ & $\hat{\sigma}_{L,d}^2$ & $\hat{\sigma}_{B,d}^2$ & $\hat{\sigma}_{J,d}^2$ \\
  \hline
 10 & Ave & 0.747 &     - &     - & 0.664 & 0.947 & 1.024 \\
   & RMSE & 0.030 &     - &     - & 0.034 & 0.028 & 0.048 \\
 20 & Ave & 0.889 &     - &     - & 0.840 & 0.972 & 0.998 \\
   & RMSE & 0.016 &     - &     - & 0.018 & 0.016 & 0.018 \\
 40 & Ave & 0.944 & 0.743 & 0.758 & 0.918 & 0.982 & 0.994 \\
   & RMSE & 0.011 & 0.020 & 0.019 & 0.011 & 0.011 & 0.011 \\
 80 & Ave & 0.976 & 0.876 & 0.884 & 0.963 & 0.995 & 1.001 \\
   & RMSE & 0.007 & 0.010 & 0.010 & 0.007 & 0.008 & 0.007 \\
160 & Ave & 0.987 & 0.937 & 0.942 & 0.980 & 0.997 & 0.999 \\
   & RMSE & 0.005 & 0.006 & 0.006 & 0.005 & 0.006 & 0.005 \\
  \hline
\end{tabular}
\caption{Ave and RMSE of all estimators introduced in Sec. \ref{sec4} for $\alpha=1, d=3$ and increasing $n_{\text{eff}}$. The entry - indicates a negative value in more than 1\% of cases.} \label{tab2}
\end{table}

\begin{table}
\centering
\begin{tabular}{rlrrrrrr}
  \hline
 $d$ &      & $\hat{\sigma}_{U,d}^2$ & $\hat{\sigma}_{N,d}^2$ & $\hat{\sigma}_{Nm,d}^2$ & $\hat{\sigma}_{L,d}^2$ & $\hat{\sigma}_{B,d}^2$ & $\hat{\sigma}_{J,d}^2$ \\
  \hline
  0 & Ave & 0.997 & 0.767 & 0.851 & 0.857 & 0.959 & 1.039 \\
   & RMSE & 0.020 & 0.024 & 0.023 & 0.021 & 0.018 & 0.020 \\
  1 & Ave & 0.994 & 0.783 & 0.844 & 0.895 & 1.007 & 1.016 \\
   & RMSE & 0.009 & 0.017 & 0.014 & 0.011 & 0.009 & 0.009 \\
  2 & Ave & 0.968 & 0.765 & 0.797 & 0.914 & 0.991 & 1.002 \\
   & RMSE & 0.010 & 0.018 & 0.017 & 0.011 & 0.011 & 0.011 \\
  3 & Ave & 0.946 & 0.744 & 0.759 & 0.920 & 0.984 & 0.997 \\
   & RMSE & 0.011 & 0.020 & 0.019 & 0.011 & 0.011 & 0.011 \\
  4 & Ave & 0.935 & 0.733 & 0.739 & 0.923 & 0.982 & 1.000 \\
   & RMSE & 0.010 & 0.020 & 0.019 & 0.011 & 0.010 & 0.010 \\
  5 & Ave & 0.931 & 0.727 & 0.730 & 0.925 & 0.980 & 1.006 \\
   & RMSE & 0.010 & 0.020 & 0.020 & 0.010 & 0.010 & 0.010 \\
   \hline
\end{tabular}
\caption{Ave and RMSE of all estimators introduced in Sec. \ref{sec4} for $\alpha=1, n_{\text{eff}}=40$ and varying $d$. The entry - indicates a negative value in more than 1\% of cases.} \label{tab3}
\end{table}

In Table \ref{tab1}, we use $n_{\text{eff}}=20, d=3$ and varying values of $\alpha$. The main findings are as follows. The Noether's estimator $\hat{\sigma}_{N,d}^2$ and its modification $\hat{\sigma}_{Nm,d}^2$ can yield negative values. If this happened in more than $1\%$ of cases, we don't report the result. For effective sample size 20, this occurred for $\alpha=1,2,5$. Hence, these estimators should not be used for small sample size. Even for $n_{\text{eff}}=80$ (results not shown), these two estimators have larger bias and RMSE compared to all other estimators, and can not be recommended.
The remaining estimators all work fine, whereby the differences for a specific estimator between the different distributions often exceed the differences between the estimators. The RMSE values are almost identical between the four estimators, and decrease in $\alpha$. The estimators $\hat{\sigma}_{U,d}^2$ and $\hat{\sigma}_{L,d}^2$ have a negative bias for all distributions for this small sample size.

In Table \ref{tab2}, we set $\alpha=1, d=3$ and vary $n_{\text{eff}}$.
As expected, bias and RMSE of all estimators tend to zero with increasing sample size; the speed of convergence of the RMSE to zero is of order $n^{-1/2}$.

Finally, Table \ref{tab3} shows the results for $\alpha=1, n_{\text{eff}}=40$  and varying values of $d$. For this sample size, the bias of $\hat{\sigma}_{L,d}^2$ is negative for all thresholds $d$, and this still holds for even larger samples. To a lesser extent, similar comments apply to $\hat{\sigma}_{U,d}^2$. The estimators $\hat{\sigma}_{B,d}^2$ and  $\hat{\sigma}_{J,d}^2$ have a smaller bias in the majority of cases, with positive or negative values depending on $d$. For $n_{\text{eff}}=80$, the last three estimators are nearly unbiased.

Summarizing the results, the estimators $\hat{\sigma}_{N,d}^2$ and $\hat{\sigma}_{Nm,d}^2$ should not be used. Since $\hat{\sigma}_{U,d}^2$ outperforms $\hat{\sigma}_{L,d}^2$ in terms of bias, not much supports the use of the latter. The coice between $\hat{\sigma}_{U,d}^2$, $\hat{\sigma}_{B,d}^2$ and $\hat{\sigma}_{J,d}^2$  is a matter of taste.
If bias is a serious concern, the last two should be preferred.
If computing time is a problem, $\hat{\sigma}_{U,d}^2$ has an advantage over $\hat{\sigma}_{J,d}^2$ and, in particular, $\hat{\sigma}_{B,d}^2$, which was computed with 999 bootstrap replications.

\subsection{Empirical coverage probability of confidence intervals for  \texorpdfstring{\boldmath $g(d)$}{g(d)}}

Here, we empirically study the coverage probabilities of confidence intervals for $g(d)$ based on the variance estimators $\hat{\sigma}_{U,d}^2, \hat{\sigma}_{L,d}^2, \hat{\sigma}_{B,d}^2$ and  $\hat{\sigma}_{J,d}^2$, using the values of $\alpha, d$ and $n_{\text{eff}}$ as in \ref{sec5-1};
hence, in this subsection, the focus is on the standard deviation instead of the variance.
Based on Theorem \ref{theorem4}, a confidence interval with asymptotic coverage probability $1-\gamma$ is given by
\begin{align*}
  \left[ \max\left\{ \hat{g}_n(d) - \frac{z_{1-\gamma/2}\, \hat\sigma_d}{ (nU_n^{(2)}(d))^{1/2}}, 0\right\},  \;
  \min\left\{ \hat{g}_n(d) + \frac{z_{1-\gamma/2}\, \hat\sigma_d}{ (nU_n^{(2)}(d))^{1/2}}, 1\right\} \right],
\end{align*}
where $z_{p}=\Phi^{-1}(p)$, and $\hat\sigma_d^2$ stands for one of the four variance estimators specified above.

The results for confidence level $1-\gamma=0.95, n_{\text{eff}}=20, d=3$ and varying $\alpha$ are given in Table \ref{tab4}. First, we note that all intervals are anticonservative, i.e. have coverage probability smaller than 0.95. Notably, the coverage probability using the first three estimators is as low as 0.90 for $\alpha=5$.
The intervals based on $\hat{\sigma}_{B,d}$ and $\hat{\sigma}_{J,d}$ behave quite similarly, the first having the edge over the second for small values of $\alpha$, and vice versa for larger values. They have slightly better empirical coverage in most cases than the interval based on $\hat{\sigma}_{U,d}$.

\begin{table}
\centering
\setlength{\tabcolsep}{12pt}
\begin{tabular}{rrrrr}
  \hline
$\alpha$ & $\hat{\sigma}_{U,d}^2$ & $\hat{\sigma}_{L,d}^2$ & $\hat{\sigma}_{B,d}^2$ & $\hat{\sigma}_{J,d}^2$ \\
  \hline
 0.2 & 91.7 & 90.4 & 94.3 & 92.7 \\
 0.5 & 92.7 & 91.7 & 94.0 & 93.6 \\
 1.0 & 92.4 & 91.7 & 93.5 & 93.7 \\
 2.0 & 91.7 & 91.1 & 92.4 & 93.2 \\
 5.0 & 89.6 & 89.4 & 90.1 & 91.0 \\
   \hline
\end{tabular}
\setlength{\tabcolsep}{6pt}
\caption{Empirical coverage probability of 0.95-confidence intervals for $g(d)$ based on different estimators for effective sample size $n_{\text{eff}}=20, d=3$ and varying $\alpha$.} \label{tab4}
\end{table}

Table \ref{tab5} shows the results for $1-\gamma=0.95, \alpha=1, d=3$ and increasing sample sizes. For $n_{\text{eff}}=40$ or larger, all intervals seem to work sufficiently well. However, a look at Table \ref{tab6}, where $n_{\text{eff}}=40$ and $\alpha=1$, shows that the empirical coverage of the interval using $\hat{\sigma}_{L,d}$ is still between 0.91 and 0.94, whereas the other intervals take values between 0.93 and 0.95.
Hence, as in subsection \ref{sec5-1}, one should choose any estimator out of $\hat{\sigma}_{U,d}^2$, $\hat{\sigma}_{B,d}^2$ and $\hat{\sigma}_{J,d}^2$ to get reliable confidence intervals.

\begin{table}
\centering
\setlength{\tabcolsep}{12pt}
\begin{tabular}{rrrrr}
  \hline
$n_{\text{eff}}$ & $\hat{\sigma}_{U,d}^2$ & $\hat{\sigma}_{L,d}^2$ & $\hat{\sigma}_{B,d}^2$ & $\hat{\sigma}_{J,d}^2$ \\
  \hline
  10 & 88.6 & 86.8 & 91.9 & 91.3 \\
  20 & 92.6 & 91.8 & 93.5 & 93.7 \\
  40 & 94.2 & 93.9 & 94.5 & 94.8 \\
  80 & 94.6 & 94.4 & 94.8 & 94.8 \\
 160 & 94.6 & 94.5 & 94.7 & 94.7 \\
  \hline
\end{tabular}
\setlength{\tabcolsep}{6pt}
\caption{Empirical coverage probability of 0.95-confidence intervals for $g(d)$ based on different estimators for $\alpha=1, d=3$ and increasing effective sample size.} \label{tab5}
\end{table}

\begin{table}
\centering
\setlength{\tabcolsep}{12pt}
\begin{tabular}{rrrrr}
  \hline
$d$ & $\hat{\sigma}_{U,d}^2$ & $\hat{\sigma}_{L,d}^2$ & $\hat{\sigma}_{B,d}^2$ & $\hat{\sigma}_{J,d}^2$ \\
  \hline
  0 & 93.1 & 91.0 & 92.9 & 93.8 \\
  1 & 94.3 & 93.1 & 94.5 & 94.6 \\
  2 & 94.2 & 93.4 & 94.5 & 94.6 \\
  3 & 94.0 & 93.6 & 94.6 & 94.7 \\
  4 & 93.8 & 93.6 & 94.2 & 94.4 \\
  5 & 93.7 & 93.6 & 94.2 & 94.4 \\
  \hline
\end{tabular}
\setlength{\tabcolsep}{6pt}
\caption{Empirical coverage probability of 0.95-confidence intervals for $g(d)$ based on different estimators for $\alpha=1, n_{\text{eff}}=40$
and varying $d$.} \label{tab6}
\end{table}

\section{Application to daily precipitation data} \label{sec6}

In this section, we apply the new tail statistic to several data sets of daily areal and point precipitation.
Establishing a probability distribution that provides a good fit to daily precipitation depths has long been a topic of interest, in particular in the areas of stochastic precipitation models, frequency analysis of precipitation and precipitation trends related to global climate change \citep{YE:2018}. Hereby, the wet-day precipitation series is the primary series considered, while a probabilistic representation of precipitation occurrences can be separately described.
A review of the literature given by \citet{YE:2018} reveals the prominent position of the gamma distribution, which was used for daily stochastic precipitation modeling already in the early 1950s \citep{TH:1951}.
In all fields mentioned above, not only the center of the distribution has to be modeled accurately, but also the distributional tail behavior is of special importance.
For the central part of the distribution of monthly or seasonal precipitation, the gamma distribution is a reasonable probability model \citep{WI:2000}; this can be different for daily precipitation or in the distributional tails. For example, \citep{YE:2018} concludes that the gamma distribution is often a reasonable model for point wet-day series in the United States. Occasionally, however, very long series are better approximated by a kappa distribution, a rather complex model with 4 parameters.

\medskip
First, we consider daily country average precipitation in Finland and Norway from 2015 to 2019, measured in centimeters. Data is available from
\url{https://www.kaggle.com/datasets/adamwurdits/finland-norway-and-sweden-weather-data-20152019}, where also additional information can be found.
Figure \ref{fig3} shows the plots of $\hat{g}_n(d)$ together with confidence intervals for confidence level 0.95, using the variance estimator $\hat{\sigma}_{U,d}^2$. The upper panel shows the graph for Finland (omitting 22 days without precipitation, the sample size is $n=1804$), the lower panel for Norway ($n=1826$).
For Finland, the plot shows a horizontal line, roughly at 0.6, corresponding to a gamma distribution with shape parameter 0.58, thus having a longer tail than the exponential distribution.
For Norway, the plot shows a horizontal line at 0.5 for values of $d$ up to 12, corresponding to an exponential distribution, but $\hat{g}_n(d)$ decreases slightly in the tail. Therefore, a gamma model for the daily precipitation in the case of Norway is questionable.

\begin{figure}
  \centering
  \includegraphics[scale=0.7]{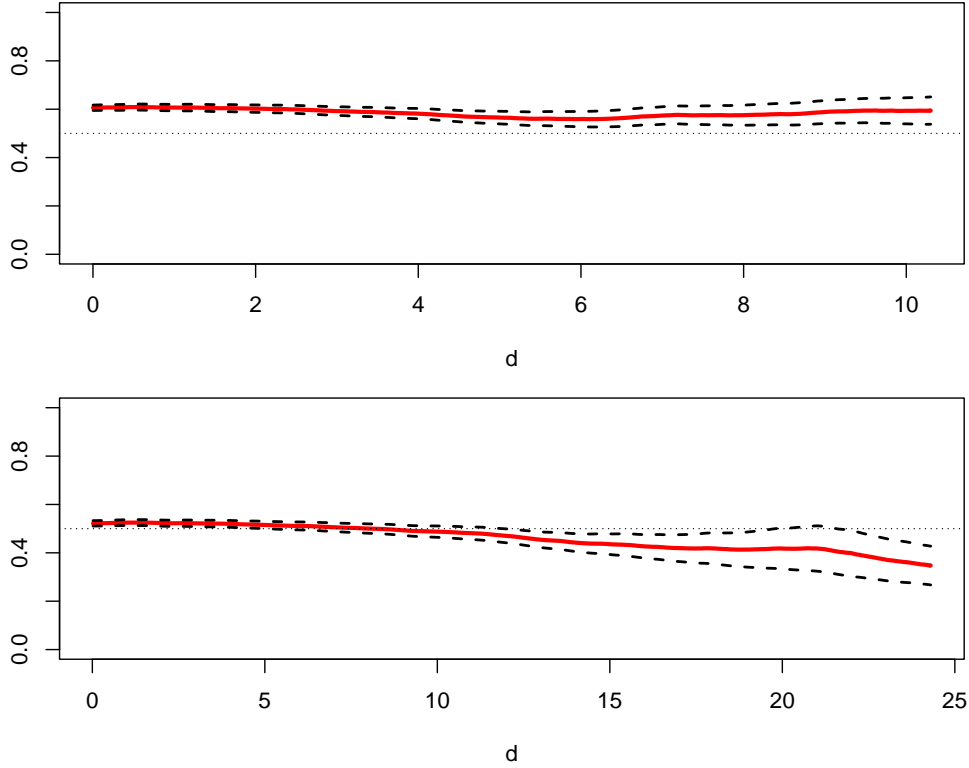}
  \caption{Graph of the estimator $\hat{g}_n(d)$ together with (pointwise) confidence limits for the daily country average precipitation in Finland (upper panel) and Norway (lower panel). \label{fig3}}
\end{figure}

\medskip
\begin{sloppypar}
Second, we analyze daily point precipitation from January 1, 2000, to December 31, 2019, at three Canadian centres, namely Calgary, Montreal and Vancouver. The datasets are subsets of longer series available under
\url{https://www.kaggle.com/datasets/aturner374/eighty-years-of-canadian-climate-data},
where further information can be found. 
The sample size, i.e. the number of wet days, is $2396, 3750$ and $3389$ for Calgary, Montreal and Vancouver, respectively.
The plot of $\hat{g}_n(d)$ with 0.95 confidence bounds for these datasets is presented in Figure \ref{fig4}. The graph for Calgary is increasing up to $d=20$; hence, a gamma distribution won't yield an adequate fit in this part of the distribution. For larger values, the graph is nearly horizontal at a value around 0.72, corresponding to a gamma distribution with shape parameter 0.31.
The graph for Montreal shows a nearly horizontal line, apart from a bend for very small values of $d$. The value of $\hat{g}_n(d)$ is 0.67 for $d=10$, which corresponds to $\alpha=0.42$.
Similarly, the graph for Vancouver is a nearly horizontal line. The value of $\hat{g}_n(10)$ is $0.56$, corresponding to $\alpha=0.73$.
Hence, for Montreal as well as Vancouver, a gamma model seems to be a good approximation in the centre and in the tail of the distribution of daily precipitation.
\end{sloppypar}

\begin{figure}
  \centering
  \includegraphics[scale=0.7]{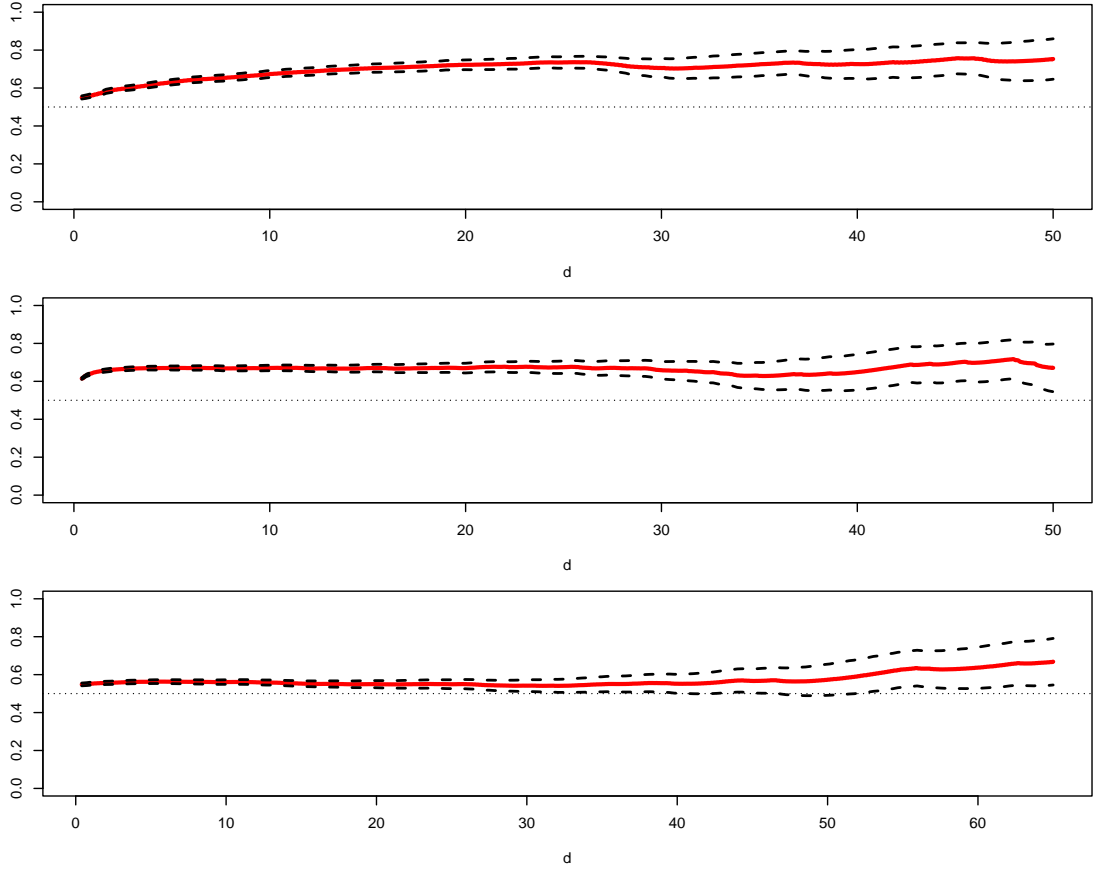}
  \caption{Graph of the estimator $\hat{g}_n(d)$ together with (pointwise) confidence limits for daily precipitation at Calgary (upper panel), Montreal (middle panel) and Vancouver (lower panel). \label{fig4}}
\end{figure}



\appendix
\section{Proofs and additional results} \label{appendix}

\begin{proposition} \label{prop}
Let $X$ and $Y$ be independent random variables, and $X \sim \Gamma(\alpha_{1}, \beta), Y \sim \Gamma(\alpha_{2}, \beta)$. Then,
\begin{eqnarray*}
\mathbb{E}\left[\frac{|X-Y|}{X+Y}\right] &=&
\dfrac{1}{2^{\alpha_{1}+\alpha_{2}}\alpha_{1} \alpha_{2}
		\mathrm{B}(\alpha_{1},\alpha_{2})}
\Big( \alpha_{1}+\alpha_{2} + (\alpha_{1}-\alpha_{2})\cdot \\
&& \quad \big( {}_{2}F_{1} \left(1,-\alpha_{1},\alpha_{2}+1;-1\right)
    - {}_{2}F_{1} \left(1,-\alpha_{2},\alpha_{1}+1;-1\right) \big)\Big), \\
\mathbb{E}\left[ \left(\frac{X-Y}{X+Y}\right)^2 \right] &=&
\left( \left(\alpha_{1}-\alpha_{2})^2\right)+\alpha_{1}+\alpha_{2} \right)
		\Gamma(\alpha_{1}+\alpha_{2})
/ \Gamma(\alpha_{1}+\alpha_{2}+2),
\end{eqnarray*}
where $\mathrm{B}(p,q)$ is the beta function, defined by
\[
\mathrm{B}(p,q) = \int_{0}^{1} x^{p-1}(1-x)^{q-1} dx,
\]
and ${}_{p} F_{q}(a_{1},\ldots,a_{p};b_{1}, \ldots,b_{q};z)$ denotes
the generalized hypergeometric function
\[
{}_{p} F_{q}(a_{1},\ldots,a_{p};b_{1}, \ldots,b_{q};z)
=\sum_{k=0}^{\infty}
\dfrac{(a_{1})_{k} \cdots (a_{p})_{k}}{(b_{1})_{k} \cdots (b_{q})_{k}}
\dfrac{z^{k}}{k!}.
\]
\end{proposition}

\begin{proof}
The densities of $X$ and $Y$ are
\[
f(x;\alpha_i,\beta) =  \dfrac{1}{\Gamma(\alpha_i)\beta^{\alpha_i}} x^{\alpha_i-1}\exp\left(-\dfrac{x}{\beta}\right), \quad \text{for }x>0,
\]
where $\alpha_i,\beta>0$ $(i=1,2)$. Then, $V=X/Y$ has a  beta prime distribution with parameters $\alpha_{1},\alpha_{2}$, which density function is given by
\[
g(v;\alpha_{1},\alpha_{2}) = \dfrac{1}{\mathrm{B}(\alpha_{1},\alpha_{2})}
v^{\alpha_{1}-1}(v+1)^{-(\alpha_{1}+\alpha_{2})}, \quad v>0.
\]
We have to evaluate the expectation \(\mathbb{E}[|X-Y|/(X+Y)]=\mathbb{E}[|V-1|/(V+1)]\).
Since
\[
\dfrac{|v-1|}{v+1}= \left|1-\dfrac{2}{v+1}\right|,
\]
we obtain
\begin{eqnarray*}
\lefteqn{\mathbb{E}[|V-1|/(V+1)|] \;=\; \mathbb{E}[|1-2/(1+V)|]} \\
&=&
\int_{0}^{1}
\left(
\dfrac{2}{v+1} -1
\right) g(v;\alpha_{1},\alpha_{2}) dv 
+ \int_{1}^{\infty}
\left(
1-\dfrac{2}{v+1}
\right) g(v;\alpha_{1},\alpha_{2}) dv.
\end{eqnarray*}
Evaluating the integrals with the software \textit{Mathematica} yields the result. An analogous computation yields the second moment.
\end{proof}

\section*{Disclosure statement}
No potential conflict of interest was reported by the authors.
\medskip

\section*{Acknowledgments}
We thank an anonymous reviewer for his constructive and helpful comments.

\bibliographystyle{apalike}

\end{document}